\documentclass{amsart}
\usepackage{amsthm,amssymb}
\usepackage{pstricks, pst-plot, pst-node, pst-grad, pst-math}
\usepackage{graphicx,color}
\usepackage{todonotes}

\newtheorem{theorem}{Theorem}[section]
\newtheorem{lemma}[theorem]{Lemma}
\newtheorem{proposition}[theorem]{Proposition}
\newtheorem{corollary}[theorem]{Corollary}
\theoremstyle{remark}

\theoremstyle{definition}
\newtheorem{definition}[theorem]{Definition}

\newtheorem{remark}[theorem]{Remark}

\newcommand{\rd}{{\rm d}}
\newcommand{\e}{{\rm e}}

\newcommand{\N}{{\mathbb N}}
\newcommand{\R}{{\mathbb R}}
\newcommand{\C}{{\mathbb C}}
\newcommand{\Z}{{\mathbb Z}}

\newcommand\beq{\begin{equation}}
\newcommand\eeq{\end{equation}}
\newcommand{\beqnt}{\begin{equation*}}
\newcommand{\eeqnt}{\end{equation*}}

\newcommand{\dist}{\mathrm{dist}}

\newcommand\I{\mathrm{i}}
\newcommand{\set}[2]{\{#1 : #2 \}}
\newcommand{\Set}[2]{\Big\{  #1 : #2 \Big\}  }

\begin{document}
\date{}
\title{Embedded eigenvalues of generalized Schr\"odinger operators}
\author{Jean-Claude Cuenin}
\address{Mathematisches Institut, Ludwig-Maximilians-Universit\"at M\"unchen, 80333 Munich, Germany}
\email{cuenin@math.lmu.de}
\begin{abstract}
We provide examples of operators $T(D)+V$ with decaying potentials that have embedded eigenvalues. The decay of the potential depends on the curvature of the Fermi surfaces of constant kinetic energy $T$. We make the connection to counterexamples in  Fourier restriction theory.
\end{abstract}
\maketitle

\section{Introduction}
The question of whether the Schr\"odinger operator 
\begin{align}\label{Schrodinger operator}
-\Delta+V
\end{align}
with a decaying potential $V$ can have eigenvalues embedded in the continuous spectrum is notoriously difficult. On physical grounds, one has good reasons to expect that the potential cannot prevent a particle from escaping to spatial infinity, due to quantum mechanical tunneling. This argument would lead one to believe that embedded eigenvalues cannot occur. It came as a big surprise in the early days of quantum mechanics when Wigner and von Neumann \cite{wigner1929merkwurdige} (see also \cite[Section~XII.13]{MR0493421} for a corrected proof) found an example of a bounded potential that decays like $1/|x|$ at infinity such that the operator \eqref{Schrodinger operator} has an embedded eigenvalue $\lambda=1$. The crucial feature of the potential is its delicate oscillatory pattern, which causes waves to be reflected coherently. On the other hand, Kato~\cite{MR0108633} proved that if $V$ is bounded and $V(x)=o(1/|x|)$ as $|x|\to\infty$, then there are no embedded eigenvalues (see  \cite{MR0276624} and \cite{MR0247300} for the case of long-range perturbations, which we will not discuss here.)

A different example leading to an embedded eigenvalue for \eqref{Schrodinger operator} was constructed by Ionescu and Jerison \cite{MR2024415}. The Ionescu-Jerison potential is bounded, decays like $1/|x|$ in one coordinate direction and like $1/|x|^2$ in the others. While the Wigner-von Neumann construction is reduced to a one-dimensional problem by spherical symmetry, the Ionescu-Jerison construction is inherently multi-dimensional. Note that the Wigner-von Neumann potential is in $L^q(\R^d)$ for any $q>d$, while the Ionescu-Jerison potential is in $L^q(\R^d)$ for any $q>(d+1)/2$. Koch and Tataru \cite{KochTataru2006} proved the absence of embedded eigenvalues for \eqref{Schrodinger operator} for a large class of potentials including $L^{(d+1)/2}(\R^d)$. This is sharp in view of the Ionescu-Jerison example.

Recently, Frank and Simon \cite{2015arXiv150401144F} adapted and simplified the two examples to show that embedded eigenvalues may occur for potentials with arbitrary small $L^q$ norms, More precisely, they proved the following. 
\begin{itemize}
\item[(A)] There exists a sequence of potentials $V_n\in L^q(\R^d)$ for any $q>(d+1)/2$ such that $-\Delta+V_n$ has eigenvalue $\lambda=1$ and $\lim_{n\to\infty}\|V_n\|_{L^q(\R^d)}= 0$.
\item[(B)] There exists a sequence of \emph{radial} potentials $V_n\in L^q(\R^d)$  for any $q>d$ such that $-\Delta+V_n$ has eigenvalue $\lambda=1$ and $\lim_{n\to\infty}\|V_n\|_{L^q(\R^d)}= 0$.
\end{itemize}  
The potentials can be chosen real-valued, so that \eqref{Schrodinger operator} is self-adjoint. Conversely, they prove that even for complex-valued potentials $V\in L^q(\R^d)$, with $d/2\leq q\leq (d+1)/2$ in the general case and $d/2\leq q\leq d$ in the radial case (in both cases $q\neq d/2$ if $d=2$), every eigenvalue $\lambda\in\C$ of \eqref{Schrodinger operator} satisfies 
\begin{align}\label{Lieb-Thirring}
|\lambda|^{1-d/(2q)}\leq C_{d,q}\|V\|_{L^q(\R^d)}.
\end{align}
Since the exponent of $|\lambda|$ is nonnegative, \eqref{Lieb-Thirring} implies the absence of eigenvalues in $(0,\infty)$ for potentials with small $L^q$ norm.
The result in the general case was proved before by Frank \cite{MR2820160} for non-embedded eigenvalues, i.e.\ $\lambda\in\C\setminus[0,\infty)$. In view of (A)--(B), the assumption on the upper bound for $q$ are optimal. Counterexamples for $q<d/2$ are of a very different nature and we will not consider them here. We just mention \cite{MR1809288} and \cite{MR1880829} where examples of unbounded, compactly supported potentials $V\in L^q(\R^d)$ (with $q<d/2$) were constructed such that \eqref{Schrodinger operator} has eigenvalue $\lambda=1$. 

There is an interesting connection between the above two examples (Wigner-von Neumann and Ionescu-Jerison) and counterexamples in the Fourier restriction problem. In its dual form, the latter asks for which exponents $p,q\in [1,\infty]$ the inequality
\begin{align}\label{restriction problem}
\|\widehat{g\rd\sigma}\|_{L^{p'}(\R^d)}\leq C\|g\|_{L^{q'}(S^{d-1})}
\end{align} 
holds for any $g\in L^{q'}(S^{d-1})$. Here, $\rd\sigma$ is the canonically induced surface measure on $S^{d-1}$. Necessary conditions for \eqref{restriction problem} to hold are 
\begin{align}\label{necessary conditions restriction}
1\leq p< \frac{2d}{d+1},\quad \frac{1}{q}\geq \frac{d+1}{d-1}\frac{1}{p'},
\end{align}
see e.g. \cite{MR1232192}. Whether these conditions are also sufficient is one of major open problems in harmonic analysis, known as the restriction conjecture. For $q=2$ it is known to be true, i.e.\ \eqref{restriction problem} holds for
\begin{align}\label{Stein Tomas exponents}
1\leq p\leq \frac{2(d+1)}{d+3}.
\end{align}
This is the content of the Stein-Tomas theorem (see \cite[Proposition IX.2.1]{MR1232192}). If one restricts the class of $g$ in \eqref{restriction problem} to radial functions, then the first condition
in \eqref{necessary conditions restriction} is in fact sufficient (see \cite[Chapter 8, Proposition 5.1]{MR2827930}). The necessity follows from the fact that
\begin{equation}\label{asymptotic surface measure}
\begin{split}
\widehat{\rd\sigma}(x)&=2\pi|x|^{-\frac{d-2}{2}}J_{\frac{d-2}{2}}(|x|)\\
&=2\sqrt{2\pi}|x|^{-\frac{d-1}{2}}\cos\left(|x|-\frac{\pi(d-1)}{4}\right)+\mathcal{O}(|x|^{-\frac{d+1}{2}})
\end{split}
\end{equation}
as $|x|\to\infty$ and by taking $g=1$ in \eqref{restriction problem}. In the nonradial case, the necessity of the second condition in \eqref{necessary conditions restriction} follows from the so-called Knapp example (see e.g.\ \cite{MR3243741} or \cite{MR3052498} for a textbook presentation). The connection to the problem of embedded eigenvalues is twofold. First, the numerology is such that 
\begin{align*}
\mathcal{B}(L^p(\R^d),L^{p'}(\R^d))&=L^d(\R^d)\quad\mbox{for}\quad p=\frac{2d}{d+1},\\
\mathcal{B}(L^p(\R^d),L^{p'}(\R^d))&=L^{\frac{d+1}{2}}(\R^d)\quad\mbox{for}\quad p=\frac{2(d+1)}{d+3}.
\end{align*}
Here $\mathcal{B}(L^p(\R^d),L^{p'}(\R^d))$ denotes the space of bounded operators from $L^p$ to $L^{p'}$. 
On the left, we have precisely the endpoints in \eqref{necessary conditions restriction}, \eqref{Stein Tomas exponents}. On the right, we have the endpoints for $q$ in (A), (B). Second, the eigenfunction used in \cite{2015arXiv150401144F} for the Wigner-von Neumann example is $u(x)=\widehat{\rd\sigma}(x)w(x)$, where $w$ is a positive radial function. This is the first (and quite obvious) connection. The eigenfunction in the simplest version (the resulting potential is complex-valued in this case) of the Ionescu-Jerison example \cite{MR2024415} has the form
\begin{align}\label{Ionescu-Jerison eigenfunction}
u(x)=\e^{\I x_d}(1+|x'|^4+|x_d|^2)^{-N/2},
\end{align}
where $x=(x',x_d)\in \R^{d-1}\times\R$ and $N>(d+1)/4$ (this ensures that $u\in L^2(\R^d)$). The parabolic scaling is characteristic for Knapp type examples. In fact, it is not difficult to show that $g=\widehat{u}$ is a
superposition of infinitely many Knapp examples (see Section \ref{Section Knapp}). 

The significance of the unit sphere in the above discussion is that it coincides with the Fermi surface for the dispersion relation $T(\xi)=|\xi|^2$ at energy $\lambda=1$.
For arbitrary smooth kinetic energies (dispersion relations) $T$, the Fermi surface at energy $\lambda$ (assumed to be a regular value of $T$) is given by
\begin{align}\label{Fermi surface}
M_{\lambda}=\set{\xi\in\R^d}{T(\xi)=\lambda}.
\end{align}
If $M_{\lambda}$ is compact and has everywhere non-vanishing Gussian curvature, the conclusion of the Stein-Tomas theorem still holds, i.e.\ \eqref{restriction problem} with $q=2$ and $p$ as in \eqref{Stein Tomas exponents} is true. Moreover, analogues of the bound \eqref{Lieb-Thirring} were proved by the author in \cite{MR3608659} for some special choices of the kinetic energy. 

The aim of this note is to construct examples of embedded eigenvalues for the generalized Schr\"odinger operators, i.e.\ operators of the form
\begin{align}\label{generalized Schrodinger operator}
T(D)+V.
\end{align} 
Here $V$ is a decaying potential and the translation-invariant kinetic energy operator $T(D)$ is a Fourier multiplier acting on Schwartz functions $f$ as
\begin{align}\label{def. T(D)}
(T(D)f)(x):=(2\pi)^{-d}\int_{\R^d}\e^{\I x\cdot\xi}T(\xi)\widehat{f}(\xi)\rd \xi.
\end{align}
The convention for the Fourier transformation used here is
\begin{align*}
\widehat{f}(\xi):=\int_{\R^d}\e^{-\I x\cdot\xi}f(x)\rd x,\quad f\in L^1(\R^d).
\end{align*}
All integrals will be over $\R^d$ unless otherwise indicated and $D_j=-\I\partial_j$ for all $j=1,\ldots,d$.

The following are the main results of this note.
\begin{theorem}[Complex-valued potentials]\label{thm. Complex-valued potentials}
Let $T:\R^d\to\C$ be smooth and polynomially bounded, and let $\lambda\in\C$ be a regular value. Then there exists a sequence of smooth potentials $V_n:\R^d\to\C$, $n\in\N$, satisfying (possibly after rotating coordinates)
\begin{align}\label{decay V nonradial}
|V_n(x)|\lesssim(n+|x_1|^2+\ldots +|x_{d-1}|^2+|x_d|)^{-1}
\end{align}
and such that $\lambda$ is an eigenvalue of $T(D)+V_n(x)$ in $L^2(\R^d)$for every $n\in\N$. In particular, for any $q>(d+1)/2$, we have that $\lim_{n\to\infty}\|V_n\|_{L^q}= 0$. 
\end{theorem}
\begin{theorem}[Complex-valued potentials II]\label{thm. Complex-valued potentials II}
Let $T:\R^d\to\R$ be smooth and polynomially bounded, and let $\lambda\in\R$ be a regular value (hence $M_{\lambda}$ is a smooth hypersurface). Assume that $M_{\lambda}$ has $k<d-1$ non-vanishing principal curvatures at some point. Then there exists a sequence of smooth potentials $V_n:\R^d\to\C$, $n\in\N$, satisfying (possibly after rotating coordinates)
\begin{align}\label{decay V nonradial k nonzero curvatures}
|V_n(x)|\lesssim(n+|x_1|^2+\ldots +|x_{k}|^2+|x_{k+1}|^3+\ldots+|x_{d-1}|^3+|x_d|)^{-1}
\end{align}
and such that $\lambda$ is an eigenvalue of $T(D)+V_n(x)$ in $L^2(\R^d)$, for all $n\in\N$. In particular, for any $q>(k+2)/2+(d-1-k)/3$, we have that $\lim_{n\to\infty}\|V_n\|_{L^q}= 0$. 
\end{theorem}
\begin{theorem}[Real-valued, radial potentials]\label{thm. Radial potentials}
Let $T:\R^d\to\R$ be a radially symmetric polynomial, and let $\lambda\in\R$ be a regular value. Assume in addition that $M_{\lambda}$ is compact. Then there exists a sequence of smooth radial potentials $V_n:\R^d\to\C$, $n\in\N$, satisfying 
\begin{align}\label{V decay Radial potentials}
|V_n(x)|\lesssim(n+|x|)^{-1}.
\end{align}
and such that $\lambda$ is an eigenvalue of $T(D)+V_n(x)$ in $L^2(\R^d)$, for all $n\in\N$. In particular, $\lim_{n\to\infty}\|V_n\|_{L^q}= 0$ for any $q>d$.
\end{theorem}

\begin{theorem}[Chandrasekhar-Herbst operator]\label{thm Chandrasekhar-Herbst}
Let $H=\sqrt{-\Delta+1}-1$ be the Chandrasekhar-Herbst operator with mass $1$, and let $\lambda>0$. Then there exists a sequence of smooth potentials $V_n:\R^d\to\R$, $n\in\N$, satisfying 
\begin{align}\label{V decay Chandrasekhar-Herbst}
|V_n(x)|\lesssim(n+|x'|^2+|x_d|)^{-1}
\end{align}
holds and such that $\lambda$ is an eigenvalue of $H+V_n$, for every $n\in\N$.
\end{theorem}
The last theorem generalizes a recent result of \cite{MR3661408} to nonradial potentials. Finally, as an example of a matrix-valued differential operator, we consider the Dirac operator (see e.g.\ \cite{MR1219537} for the definition of the Dirac matrices $\alpha_j$). 
\begin{theorem}[Dirac operator]\label{thm. Dirac}
Let $\mathcal{D}=\alpha\cdot D+\beta$ be the Dirac operator with mass~$1$ and let $\lambda\in (-\infty-1)\cup (1,\infty)$. There exists a sequence of smooth (anti-hermitian) potentials $V_n$, $n\in\N$, satisfying 
\begin{align}\label{V decay Dirac}
|V_n(x)|\lesssim(n+|x'|^2+|x_d|)^{-1}
\end{align}
and such that $\lambda$ is an eigenvalue of $\mathcal{D}+V_n$ for every $n\in\N$.
\end{theorem}

\begin{remark}
All of the above theorems also hold when $\lambda$ is a critical value of~$T$. In fact, this case is easier and somewhat less interesting since the distiction between the radial and the nonradial case disappears. An inspection of the proof of Theorem \ref{thm. Complex-valued potentials} will show that if $\lambda$ is a critical value, then one can even achieve the decay $$|V_n(x)|\lesssim (n+|x|)^{-2}.$$
\end{remark}

\begin{remark}
Using \cite[Lemma 3.3]{MR3608659} together with a Birman-Schwinger argument one can show that if $T:\R^d\to\R$ is smooth and proper (i.e.\ preimages of compact sets are compact), and if $M_{\lambda}$ has everywhere nonvanishing Gaussian curvature, then $\lambda$ cannot be an eigenvalue of $T(D)+V$, provided $\|V\|_{L^{(d+1)/2}}+\|V\|_{L^{\infty}}$ is sufficiently small. 
\end{remark}

\begin{remark}
Similarly, one can show that if $M_{\lambda}$ has at least $k$ nonvanishing principal curvatures at every point, then $\lambda$ cannot be an eigenvalue, provided $\|V\|_{L^{(k+2)/2}}+\|V\|_{L^{\infty}}$ is sufficiently small. Note that there is a gap between the $q$ in Theorem \ref{thm. Complex-valued potentials II} and $q=~(k+2)/2$ here. The reason is that the latter does not consider worst-case scenarios. For instance, if $T(\xi)=\xi_1^2+\ldots+\xi_{k+1}^2$ where $k<d-1$ (i.e.\ $T$ depends on less than $d$ variables), then, e.g.\ for $\lambda=1$, the Fermi surface $M_{\lambda}$ is a cylinder, and the proof of Thereom \ref{thm. Complex-valued potentials II} shows that one can get
\begin{align*}
|V_n(x)|\lesssim(n+|x_1|^2+\ldots +|x_{k}|^2+|x_{k+1}|^m+\ldots+|x_{d-1}|^m+|x_d|)^{-1}
\end{align*}
for arbitrary positive $m$. Then $V_n\in L^q(\R^n)$ for all $q>(k+2)/2+(d-1-k)/m$. This shows that the exponent $q=(k+2)/2$ is in general optimal.
\end{remark}

\begin{remark}
By a tranference-type argument one can prove analogues of Theorems \ref{thm. Complex-valued potentials}--\ref{thm. Radial potentials} to discrete Schr\"odinger operators on $\ell^2(\Z^d)$. In this case $T$ is a function on the torus $[0,2\pi]^d$ and acts as (by slight abuse of notation)
\begin{align}\label{def. T(D) discrete case}
(T(D)f)(n):=(2\pi)^{-d}\int_{[0,2\pi]^d}\e^{\I x\cdot\xi}T(\xi)\widehat{f}(\xi)\rd \xi,\quad n\in\Z^d,
\end{align}
on functions $f:\Z^d\to\C$ with rapid decay. Here,
\begin{align*}
\widehat{f}(\xi):=\sum_{n\in\Z^d}\e^{-\I n\cdot\xi}f(n)
\end{align*}
is the discrete Fourier transform of $f$. An argument involving the Poisson summation formula shows that if $f|_{\Z^d}$ is the restriction of a smooth function $f:\R^d\to\C$ with sufficient decay, then
\begin{align}\label{Poisson summation}
(T(D)f|_{\Z^d})(n)=(T_{\rm per}(D)f)(n),\quad n\in\Z^d.
\end{align}
Here $T_{\rm per}:\R^d\to\C$ is the periodization of $T:[0,2\pi]^d\to\C$. On the left hand side $T(D)$ is understood as a multiplier on $\Z^d$ as in \eqref{def. T(D) discrete case}, while on the right hand side $T_{\rm per}(D)$ is a multiplier on $\R^d$ in the sense of \eqref{def. T(D)}. Since Theorems \ref{thm. Complex-valued potentials}--\ref{thm. Radial potentials} can be applied to $T_{\rm per}$ and the eigenfunctions in the proofs may be chosen smooth and with arbitrary decay, the discrete analogues follow directly from \eqref{Poisson summation}.
\end{remark}

\begin{remark}
An interesting example where the discrete analogue of Theorem \ref{thm. Complex-valued potentials II} applies is the (standard) discrete Schr\"odinger operator
\begin{align*}
Hf(n):=-\frac{1}{4}\sum_{j=1}^d\left(f(n+e_j)-f(n-e_j)\right)+\frac{d}{2}f(n)+V(n)f(n)
\end{align*}
on $\ell^2(\Z^d)$. The symbol of the kinetic energy is 
\begin{align*}
T(\xi)=d-\sum_{j=1}^d\cos(\xi_j),
\end{align*}
and is well known (see e.g.\ \cite[Section 4.2]{MR3449177}) that in $d\geq 3$ dimensions its Fermi surface $M_{\lambda}$ has points of vanishing Gaussian curvature, i.e.\ $k<d-1$ in Theorem~\ref{thm. Complex-valued potentials II}. 
\end{remark}

We briefly sketch the outline of this note. In Section 2 we prove Theorems~\ref{thm. Complex-valued potentials} and Theorem \ref{thm. Complex-valued potentials II} as well as some basic, but useful auxiliary results on mapping properties of pseudodifferential operators between spaces of functions with (anisotropic) symbol-type behavior. In Section 3 we connect the construction of the eigenfunctions used in the previous section to the Knapp example in harmonic analysis. In Section 4 we show that the potentials in Theorems \ref{thm. Complex-valued potentials}--\ref{thm. Complex-valued potentials II} may be chosen real-valued under additional assumptions on the kinetic energy (see Proposition~\ref{proposition real potential} and Corollary~\ref{corollary real valued V}). In Section 5 we prove Theorems \ref{thm. Radial potentials}--\ref{thm. Dirac}.  
\section{Proofs of Theorems \ref{thm. Complex-valued potentials}--\ref{thm. Complex-valued potentials II}}

\subsection{Anisotropic symbols}
Let $\rho:\R^d\to[1,\infty)$ be a temperate weight, i.e.\ a continuous function with the property that there exists $s>0$ such that for every $x,y\in\R^d$
\begin{align}\label{def. temperate weights}
\rho(x)\lesssim \rho(y)\langle x-y\rangle^{s}.
\end{align} 
Note that this implies in particular that
\begin{align}
\langle x-y\rangle^{-s}\lesssim \rho(x)\lesssim \langle x-y\rangle^{s}.
\end{align}
We fix $\gamma=(\gamma_1,\ldots,\gamma_d)\in(0,1]^d$ such that
\begin{align}\label{special case wj powers of a fixed weight bis}
\gamma_d\geq \gamma_{d-1}\geq \ldots\geq \gamma_1.
\end{align}
\begin{definition}
Let $\ell\in\R$ and assume that $\rho,\gamma$ satisfy \eqref{def. temperate weights}, \eqref{special case wj powers of a fixed weight bis}. We say that $f\in S^{\ell}_{\gamma}(\R^d)$ if $f\in C^{\infty}(\R^d)$ and for any $\alpha\in\N_0^d$ there exist constants $C_{\alpha}>0$ such that for all $x\in\R^d$
\begin{align}\label{def. Sw}
|\partial^{\alpha}f(x)|\leq C_{\alpha}\rho(x)^{\ell}(x)\rho(x)^{-\gamma_1\alpha_1}\ldots \rho(x)^{-\gamma_d\alpha_d}.
\end{align}
\end{definition}
We recall that $a\in C_{\rm pol}^{\infty}(\R^d)$, the space of smooth,  polynomially bounded functions, if $a\in C^{\infty}(\R^d)$ and for any $\alpha\in\N_0^d$, there exist constants $C_{\alpha}>0$, $m_{\alpha}\in\R$ such that for all $x\in\R^d$
\begin{align}\label{def. smooth, polynomially bounded}
|\partial^{\alpha}a(\xi)|\leq C_{\alpha} \langle\xi\rangle^{m_{\alpha}}.
\end{align}
Both $S^{\ell}_{\gamma}(\R^d)$ and $C_{\rm pol}^{\infty}(\R^d)$ are locally convex spaces when equipped with the seminorms
\begin{align*}
\|f\|_{n}&:=\sup_{|\alpha|\leq n}\sup_{x\in\R^d}|\rho(x)^{-\ell}\rho(x)^{\gamma_1\alpha_1}\ldots \rho(x)^{\gamma_d\alpha_d}f(x)|,\quad f\in S^{\ell}_{\gamma}(\R^d),\\
\|a\|_{n,m}&:=\sup_{|\alpha|\leq n}\sup_{x\in\R^d}|\langle \xi \rangle^{-m}\partial^{\alpha}a(\xi)|,\quad a\in C_{\rm pol}^{\infty}(\R^d),
\end{align*}
respectively.
Let $h\in (0,1]$ and $\xi\in\R^d$. By abuse of notation we write
\begin{align}\label{def. h times xi}  
h\xi:=(h^{\gamma_1}\xi_1,\ldots,h^{\gamma_d}\xi_d)\in\R^d.
\end{align}
Here $h$ plays the role of a semiclassical parameter.
Given $a\in C_{\rm pol}^{\infty}(\R^d)$ and $f\in S^{\ell}_{\gamma}(\R^d)$,
we define
\begin{align}\label{def. semiclassical psdo}
(a(hD)f)(x):=(2\pi)^{-d}\lim_{\epsilon\to 0}\int\int\e^{\I (x-y)\cdot\xi}\varphi(\epsilon\xi)\varphi(\epsilon y) a(h\xi)f(y)\rd y\rd\xi
\end{align}
where $\varphi\in\mathcal{S}(\R^d)$ is fixed and satisfies $\varphi(0)=1$. 
Note that the integral in \eqref{def. semiclassical psdo} is absolutely convergent, and an integration by parts argument (as in the proof of Proposition \ref{proposition pseudodifferential mapping properties} below) shows that the definition is independent of the choice of $\varphi$. For $h=1$ the definition \eqref{def. semiclassical psdo} coincides with the one of \eqref{def. T(D)}. We write it in this way to avoid the semiclassical Fourier transformation. We will set $h=1/n$ later on. Alternatively, we could stick to $h=1$, but then we will have to consider functions 
depending on $n$ in the subsequent proofs. Both points of views are equivalent, and we freely swith between one and the other in later sections.

\begin{proposition}\label{proposition pseudodifferential mapping properties}
Let $a\in C_{\rm pol}^{\infty}(\R^d)$
Then 
\begin{align}\label{eq. mapping properties a(x,D)}
a(hD):S^{\ell}_{\gamma}(\R^d)\to S^{\ell}_{\gamma}(\R^d)
\end{align}
as a continuous map, with seminorm bounds independent of $h\in (0,1]$.
\end{proposition}
\begin{proof}
We prove that $g=a(hD)f$ satisfies \eqref{def. Sw} for $\alpha=0$. The proof of the general case is similar. For $f\in S^{\ell}_{\gamma}(\R^d)$ and $\epsilon\in (0,1]$ let
\begin{align*}
g_{\epsilon}(x):=
(2\pi)^{-d}\int\int\e^{\I (x-y)\cdot\xi}\varphi(\epsilon y)\varphi(\epsilon\xi) a(h\xi)f(y)\rd y\rd \xi.
\end{align*}
Integration by parts yields
\begin{align*}
g_{\epsilon}(x)=(2\pi)^{-d}\int\int\e^{\I(x-y)\cdot\xi}L_1^{l_1}L_2^{l_2}\varphi(\epsilon\xi)\varphi(\epsilon y)a(h\xi) f(y)\rd y\rd \xi
\end{align*}
where $l_1,l_2\in\N$ and
\begin{align*}
L_1=\frac{1+(x-y)\cdot D_{\xi}}{1+|x-y|^2},\quad L_2=\frac{1-\xi\cdot  D_y}{1+|\xi|^2}.
\end{align*}
Hence, there exists $m(l_1)>0$ such that
\begin{align*}
|g_{\epsilon}(x)|\lesssim\|f\|_{l_2}\|a\|_{l_1,m(l_1)} \int\int\langle x-y\rangle^{-l_1}\langle \xi\rangle^{-l_2+m(l_1)} \rho(y)^{\ell}\rd y\rd\xi.
\end{align*}
By changing the roles of $x$ and $y$ in \eqref{def. temperate weights} we see that $\rho(y)\lesssim \rho(x)\langle x-y\rangle^s$. Choosing first $l_1$ so large that $-l_1+ls<-d$ and then $l_2$ so large that $-l_2+m(l_1)<-d$, we get
\begin{align*}
|g_{\epsilon}(x)|\lesssim \rho(x)^{\ell}
\end{align*}
where the implicit constant is independent of $\epsilon$ and $h$. The claim follows by letting $\epsilon\to 0$.
\end{proof}
\begin{definition}
For a smooth function $a:\R^d\to\C$ the \emph{Taylor support} of $a$ at the origin is the set
\begin{align*}
\mathcal{T}(a):=\set{\alpha\in\N_0^d}{\partial^{\alpha}a(0)\neq 0}.
\end{align*}
The Newton polyhedron of $a$, denoted by $\mathcal{N}(a)$, is the convex hull of the set $\set{\alpha+\R_+^d}{\alpha\in \mathcal{T}(a)}$. The set of vertices of $\mathcal{N}(a)$ is denoted by $\mathcal{V}(a)$ and is a subset of $\tau(a)$.
\end{definition}
\begin{proposition}\label{proposition symbols}
Let $a\in C_{\rm pol}^{\infty}(\R^d)$, $f\in S^{\ell}_{\gamma}(\R^d)$ and assume that \eqref{special case wj powers of a fixed weight bis} holds. Moreover, assume that 
\begin{align}\label{assumption Taylor support of a}
\emptyset\neq \tau(a)\subset\Set{\alpha\in\N_0^d}{\sum_{j=1}^d \gamma_j\alpha_j\geq 1}.
\end{align}
Then $a(hD)f\in h  S^{\ell-1}_{\gamma}(\R^d)$.
\end{proposition}

\begin{proof}
Choose $k\in\N$ such that $k\gamma_1\geq 1$. By Taylor's theorem, we have
\begin{align*}
a(\xi)=\sum_{|\alpha|\leq k}\frac{1}{\alpha!}\partial^{\alpha}a(0)\xi^{\alpha}+\sum_{|\alpha|=k}\xi^{\alpha}b_{\alpha}(\xi)
\end{align*}
where
\begin{align*}
b_{\alpha}(\xi)=k\sum_{|\gamma|=k}\frac{\xi^{\gamma}}{\gamma!}\int_0^1(1-\theta)^{k-1}\partial^{\gamma}a(\theta\xi)\rd\theta.
\end{align*}
Clearly $b_{\alpha}\in C_{\rm}^{\infty}(\R^d)$. By Proposition \ref{proposition pseudodifferential mapping properties}, we thus have that $b_{\alpha}(hD)=\mathcal{O}_{S^{\ell}_{\gamma}\to S^{\ell}_{\gamma}}(1)$. By \eqref{special case wj powers of a fixed weight bis} and the choice of $k$, it follows that 
\begin{align*}
\big|\sum_{|\alpha|=k}(hD)^{\alpha}b_{\alpha}(hD)f(x)\big|\lesssim h \rho(x)^{\ell-1}.
\end{align*}
Define 
\begin{align*}
\mathcal{V}_k(a):=\mathcal{V}(a)\cap \set{\alpha\in \N_0^d}{|\alpha|\leq k}\subset\tau(a).
\end{align*}
Then we can write 
\begin{align*}
\sum_{|\alpha|\leq k}\frac{1}{\alpha!}\partial^{\alpha}a(0)\xi^{\alpha}=\sum_{\alpha\in \mathcal{V}_k(a)}\xi^{\alpha}p_{\alpha}(\xi)
\end{align*}
where $p_{\alpha}$ are polynomials. In particular, $p_{\alpha}(hD)=\mathcal{O}_{S^{\ell}_{\gamma}\to S^{\ell}_{\gamma}}(1)$, so that \eqref{assumption Taylor support of a} implies
\begin{align*}
\big|\sum_{\alpha\in \mathcal{V}_k(a)}(hD)^{\alpha}p_{\alpha}(hD)\big|\lesssim h \rho(x)^{\ell-1}.
\end{align*}
The claim is proved.
\end{proof}

\begin{proof}[Proof of Theorem \ref{thm. Complex-valued potentials}]
Let $\eta\in M_{\lambda}$. Since $\lambda$ is a regular value, we have $\nabla T(\eta)\neq 0$. By rotating coordinates if necessary, we may assume that $\nabla T(\eta)=|\nabla T(\eta)|e_d$. We fix 
\begin{align}\label{gamma, rho thm 1}
\rho(x):=(1+|x'|^4+|x_d|^2)^{1/2},\quad \gamma=(1/2,\ldots,1/2,1)
\end{align}
and define 
\begin{align}\label{Ionescu-Jerison type example}
\psi(x):=\rho(x)^{-N},\quad u(x):=\e^{\I \eta\cdot x}\psi(x),
\end{align}
where $N>(d+1)/4$. A change of variables shows that
\begin{align*}
\int_{\R^d}|u(x)|^2\rd x=\left(\int_{\R^{d-1}}(1+|x'|^4)^{-N}\rd x'\right)\left(\int_{\R}(1+|x_d|^2)^{-N+\frac{d-1}{4}}\rd x_d\right)<\infty.
\end{align*}
Proposition \ref{proposition pseudodifferential mapping properties} then implies that $T(D)u\in L^2(\R^d)$ as well. Hence $u$ is in the domain of $T(D)$. Note that since
\begin{align}\label{commutation T(D) with plane wave}
T(D)\e^{\I \eta\cdot x}=\e^{\I \eta\cdot x}T(D+\eta)
\end{align} 
as operators on smooth functions with bounded derivatives, we may assume without loss of generality that $\eta=0$. For $n\geq 1$ let $h=1/n$ and $u_n(x):=h^Nu(hx)$, where $hx$ is defined as in \ref{def. h times xi}. If $V_n(x)=W_h(hx)$, then 
\begin{align*}
(T(D)-\lambda+V_n)u_n=0\iff (T(hD)-\lambda+W_h)u=0.
\end{align*}
Since $u$ has no zeros, we may write the second equation as $W_h=-(T(hD)-\lambda)u/u$.  
We will prove that
\begin{align}\label{estimate to prove in thm. Complex-valued potentials}
W_h(x)|\lesssim h(1+|x'|^2+|x_d|)^{-1}.
\end{align}
This is equivalent to \eqref{decay V nonradial}. Since \eqref{estimate to prove in thm. Complex-valued potentials} implies that $W_h\in L^q(\R^d)$ for $q>(d+1)/2$, it also follows that
\begin{align*}
\|V_n\|_{L^q(\R^d)}=h^{-\frac{d+1}{2q}}\|W_h\|_{L^q(\R^d)}=\mathcal{O}(h^{1-\frac{d+1}{2q}})\to 0
\end{align*}
as $h\to 0$. 
To prove \eqref{estimate to prove in thm. Complex-valued potentials}
we will apply Proposition \ref{proposition symbols} to $a(\xi)=T(\xi)-\lambda$ and $f=\psi$. It easy to check that $\psi\in S^{-N}_{\gamma}(\R^d)$ if $\gamma,\rho$ are defined as in \eqref{gamma, rho thm 1}. Moreover, 
condition \eqref{assumption Taylor support of a} is clearly satisfied. Proposition \ref{proposition symbols} now yields that $W_h\in  h  S^{-1}_{\gamma}(\R^d)$; in particular, \eqref{estimate to prove in thm. Complex-valued potentials} holds.
\end{proof}
\begin{proof}[Proof of theorem \ref{thm. Complex-valued potentials II}]
Let $\eta\in M_{\lambda}$ be such that only $k<d-1$ principal curvatures are nonzero at this point. We can again assume that $\eta$ is the origin and that $\nabla T(\eta)=|\nabla T(\eta)|e_d$. Write $x=(x',x'',x_d)\in\R^k\times\R^{d-1-k}\times\R$ and fix
\begin{align}\label{gamma, rho thm 2}
\rho(x):=(1+|x'|^4+|x''|^6+|x_d|^2)^{1/2},\quad \gamma=(\underbrace{1/2,\ldots,1/2}_{k \mbox{ times }},\underbrace{1/3,\ldots,1/3}_{d-1-k \mbox{ times }},1).
\end{align}
We then define $u$ as in \eqref{Ionescu-Jerison type example}, but with $\rho$ as in \eqref{gamma, rho thm 2}.
Near the origin, $M_{\lambda}$ is the graph of a smooth function $\Phi:\R^{k}\times\R^{d-1-k}\to\R$ over the last coordinate axis, and $\Phi(0)=0$, $\nabla\Phi(0)=0$. By the curvature assumption, we can write (after perhaps a linear change of the $(\xi',\xi'')$ coordinates)
\begin{align*}
\Phi(\xi',\xi'')=\frac{1}{2}\sum_{j=1}^{k}s_j\xi_j^2+\mathcal{O}(|(\xi',\xi'')|^3),\quad s_j=\pm 1,
\end{align*} 
as $(\xi',\xi'')\to 0$, with similar estimates for all derivatives.
By differentiating the identity $T(\xi',\xi'',\Phi(\xi',\xi''))=\lambda$ twice, we find that
\begin{align*}
\partial_{i}\partial_j T(0)=-(\partial_d T(0))^{-1}\partial_{i}\partial_j\Phi(0),\quad i,j\leq d-1,
\end{align*} 
and hence the Taylor support of $T(\xi)-\lambda$ at $\xi=0$ is contained in the set
\begin{align*}
\Set{\alpha\in\N_0^d}{\frac{1}{2}(\alpha_1+\ldots+\alpha_{k})+\frac{1}{3}(\alpha_{k+1}+\ldots+\alpha_{d-1})+\alpha_d\geq 1}.
\end{align*}
We define $W_h:=-(T(hD)-\lambda)u/u$ and $V_n(x):=W_h(hx)$ for $h=1/n$.
As in the proof of Theorem~\ref{thm. Complex-valued potentials} it is sufficient to prove that
\begin{align}\label{estimate to prove in thm. Complex-valued potentials II}
|W_h(x)|\lesssim h(1+|x'|^2+|x''|^3+|x_d|)^{-1}.
\end{align}
Indeed, \eqref{estimate to prove in thm. Complex-valued potentials II} is equivalent to \eqref{decay V nonradial k nonzero curvatures}, and for any $q>(k+2)/2+(d-1-k)/3$, we have that
\begin{align*}
\|V_n\|_{L^q(\R^d)}=h^{-\frac{1}{q}-\frac{k}{2q}-\frac{d-k-1}{3q}}\|W_h\|_{L^q(\R^d)}=\mathcal{O}(h^{1-\frac{1}{q}-\frac{k}{2q}-\frac{d-k-1}{3q}})\to 0
\end{align*}
as $h\to 0$. The bound \eqref{estimate to prove in thm. Complex-valued potentials II} follows from Proposition \ref{proposition symbols}.
\end{proof}

\section{Connection to Knapp's example}\label{Section Knapp}
Let us now quickly recall Knapp's homogeneity argument (see e.g.\ \cite[Example~ 10.4.4]{MR3243741}) for the sphere $S^{d-1}$. The same argument yields a necessary condition for the restriction problem of $M_{\lambda}$ if the latter has everywhere nonzero Gaussian curvature. For the present purpose it is more convient to consider the restriction problem in its original (not adjoint) form
\begin{align}\label{original restriction problem}
\|\widehat{f}\|_{L^q(S^{d-1})}\leq C\|f\|_{L^p(\R^d)},\quad \mbox{for all  }f\in\mathcal{S}(\R^d).
\end{align}
This is equivalent to \eqref{restriction problem}. Since the problem is translation-invariant, we may assume that the north pole of the unit sphere is the origin of our coordinate system. Let $f=\chi_R$ be a smoothed version of the characteristic function for a rectangle $R$. Then for $0<\delta<1$ and $f_{\delta}(x):=f(\delta x',\delta^2 x_d)$ one easiy checks that
\begin{align*}
\|f_{\delta}\|_{L^p(R^{d})}\sim \delta^{-\frac{d+1}{p}},\quad 
\|\widehat{f}_{\delta}\|_{L^q(S^{d-1})}\sim \delta^{\frac{d-1}{q}-d-1}.
\end{align*}
Since $\delta<1$, this shows that \eqref{original restriction problem} can only hold if the second condition in \eqref{necessary conditions restriction} is satisfied.

It is customary to call $f_{\delta}$ a Knapp example. 
If we look at its phase-space portrait we see that it is supported on a rectangle $R_{\delta}$ of size $\delta\times \delta^{-2}$ and that its Fourier transform decays rapidly off the dual rectangle $$R_{\delta}^*:=\set{\xi\in\R^d}{|\xi|'\lesssim\delta,\, |\xi_d|\lesssim \delta^{2}}.$$ Since  
\begin{align*}
T(\xi)=\lambda+(\partial_1T(0))\xi_d+\mathcal{O}(|\xi|^2)
\end{align*}
as $|\xi|\to 0$, we thus have
\begin{align*}
T(D)f_{\delta}(x)\sim T|_{R_{\delta}^*}f_{\delta}(x)\sim \lambda f_{\delta}(x)+\mathcal{O}(\delta^2).
\end{align*}
Since pseudodifferential operators do not move the support too much, this function is essentially supported (up to rapidly decaying tails) on the double of $R_{\delta}$. Cover $\R^d\setminus B(0,1)$ by finitely overlapping rectangles $R_j$ of size $2^j\times 4^j$, such that for $j\in \N$,
\begin{align*}
x\in R_j\implies |x'|^2+|x_d|\sim 4^j.
\end{align*} 
This can be done by a slight modification of the standard Littlewood-Paley decomposition. Let $\{\chi_j\}_{j=1}^{\infty}$ be a partition of unity subordinate to this cover and define
\begin{align}\label{superposition of Knapp examples}
u(x):=\sum_{j=1}^{\infty} 4^{-Nj} \chi_{j}(x),\quad x\in \R^d\setminus B(0,1).
\end{align}
The function $u$ may be viewed as a superposition of Knapp examples $\{\chi_{j}\}_{j=1}^{\infty}$ corresponding to rectangles $\{R_{j}\}_{j=1}^{\infty}$.
It is not difficult to show that, for $x$ in the exterior of the unit ball, we have that
\begin{align}\label{connection Ionescu-Jerison to Knapp}
u(x)\sim (|x'|^2+|x_d|)^{-N}
\end{align}
as well as
\begin{align*}
T(D)u(x)\sim \sum_{j=1}^{\infty}4^{-Nj} T|_{R_{j}^*}\chi_{R_j}(x)\sim \lambda u(x)+\mathcal{O}((|x'|^2+|x_d|)^{-N-1}).
\end{align*}
Modifying $u$ in \eqref{superposition of Knapp examples} by a compactly supported function, we easily arrange that $u>0$ and that $V=-(T(D)-\lambda)u/u$ is smooth and satisfies \eqref{decay V nonradial}. The relation~\eqref{connection Ionescu-Jerison to Knapp} connects the Knapp example and the Ionescu-Jerison example (i.e.\ the eigenfunction used in the proof of Theorem \ref{thm. Complex-valued potentials}). There is a similar connection between a Knapp example for surfaces with $k$ non-vanishing principal curvatures and the eigenfunction used in the proof of Theorem~\ref{thm. Complex-valued potentials II}.

\section{Real-valued potentials}\label{Section Real-valued potentials}

\begin{lemma}\label{lemma time reversal symmetry}
Let $T\in C_{\rm pol}^{\infty}(\R^d)$ be real-valued and time-reversal symmetric, i.e.\ $T(\xi)=T(-\xi)$. Let $\varphi(x)=\sin(\eta\cdot x)\psi(x)$ where $\psi\in C_b^{\infty}(\R^d)$ is any real-valued function and $\eta\in\R^d$. Then the function $T(D)\varphi$ is real-valued.
\end{lemma}

\begin{proof}
A straightforward calculation using \eqref{commutation T(D) with plane wave} yields
\begin{align*}
T(D)\varphi(x)&=\frac{1}{2\I}\left(\e^{\I\eta\cdot x}T(D+\eta)-\e^{-\I\eta\cdot x}T(D-\eta)\right)\psi(x).
\end{align*}
Changing variables $\xi\to-\xi$ in the resulting integrals and using the fact that $\widehat{\psi}(\xi)=\widehat{\psi}(-\xi)$, one finds that $\overline{T(D)\varphi(x)}=T(D)\varphi(x)$.
\end{proof}

\begin{proposition}\label{proposition real potential}
Assume that the assumptions of Theorem~\ref{thm. Complex-valued potentials} or Theorem~\ref{thm. Complex-valued potentials II} hold and that $T$ is a polynomial such that $T(\xi)=T(-\xi)$ for all $\xi\in\R^d$. In addition, assume that, possibly after a linear bijection $L:\R^d\to\R^d$, we have $e_1\in M_{\lambda}$ and 
\begin{align}\label{technical condition real potentials}
\tau(T(e_1+\cdot)\cap\set{\alpha\in \N_0^d}{|\alpha|\in 2\N-1\mbox{ and }\alpha_1=m}=\{m e_1\}
\end{align}
for some odd integer $m$. Then the conclusion of Theorem \ref{thm. Complex-valued potentials} or Theorem \ref{thm. Complex-valued potentials II} holds with real-valued potentials. 
\end{proposition}

\begin{proof}
We first note that the assumptions on $T$ are invariant under linear bijections $L:\R^d\to\R^d$. Moreover, if $u$ is a real-valued function, then so is $u_L(x):=u(^tLx)$.  
Since
\begin{align*}
\widehat{u_L}(\xi)=|\det L|^{-1}\widehat{u}(L^{-1}\xi),
\end{align*}
we have that $T(D)u_L(x)=T_L(D)u(^tLx)$, where $T_L(\xi)=T(L\xi)$. If $\eta\in M_{\lambda}\setminus\{0\}$ we can arrange that $Le_1=\eta$. Since $T(\xi)=T(-\xi)$, the origin is a critical point, so there is no loss of generality in omitting it.

Assume now that the assumptions of the proposition hold with $L$ being the identity. We will prove the claim under the assumptions of Theorem \ref{thm. Complex-valued potentials}. The case of Theorem \ref{thm. Complex-valued potentials II} is analogous. We also assume that $n=1$ since the general case follows by scaling as in the proof of Theorem \ref{thm. Complex-valued potentials}.

We would like to choose the eigenfunction as in Lemma \ref{lemma time reversal symmetry}, i.e.
\begin{align*}
\widetilde{u}(x):=\sin(x_1)\psi(x)
\end{align*}
where $\psi$ is given by
\begin{align*}
\psi(x):=(1+|P_{\nu}x|^2+|P_{\nu}^{\perp}x|^4)^{-N/2}
\end{align*}
with $N>d/2$. Here $P_{\nu}$ is the orthogonal projection onto the subspace spanned by the unit vector $\nu=\nabla T(e_1)/|\nabla T(e_1)|$ and $P_{\nu}^{\perp}=\rm{id}-P_{\nu}$.
The problem is thatv$u$ now vanishes on the hypersurfaces $x_1=k\pi$, $k\in\Z$, so we cannot solve $(T(D)-\lambda+V)\widetilde{u}=\widetilde{u}$ for $V$ by dividing with $\widetilde{u}$. We follow the basic strategy of Ionescu and Jerison \cite{MR2024415} to remove the singularities. We are looking for a function $u$ with the same zero set (including multiplicities) as $\widetilde{u}$ such that $(T(D)-\lambda) u=0$ on this set.
We first compute
\begin{align}\label{T(D) applied to sin}
T(D)\widetilde{u}(x)&=\frac{1}{2\I}\e^{\I x_1}(T(D+e_1)-T(D-e_1))\psi(x)+\sin(x_1)T(D-e_1)\psi(x).
\end{align}
Set $f=(T(D)-\lambda)\widetilde{u}$. As in the proof of Theorem \ref{thm. Complex-valued potentials} we use Proposition \ref{proposition symbols} to infer that
\begin{align}\label{eq. derivative bounds f}
|\partial^{\alpha} f(x)|\leq C (1+|P_{\nu}x|+|P_{\nu}^{\perp}x|^2)^{-1}\psi(x),\quad |\alpha|\leq m.
\end{align}
In fact, $a(\xi)=T(\xi+e_1)-T(\xi-e_1)$ as well as $a(\xi)=T(\xi-e_1)-\lambda$ satisfy the assumptions of Proposition \ref{proposition symbols} (after relabelling coordinates). For example, in the first instance, $a(0)=T(e_1)-T(-e_1)=0$ and $\nabla a(0)=\nabla T(e_1)-\nabla T(-e_1)=2\nabla T(e_1)$. Here we used the time-reversal symmetry $T(\xi)=T(-\xi)$ and its consequence $\nabla T(\xi)=-\nabla T(-\xi)$.
Define
\begin{align}\label{def. w}
w(x):=\sum_{k\in\Z}\cos(k\pi)(x_1-k\pi)^m\chi(x_1-k\pi)f(k\pi,x')
\end{align} 
where $x'=(x_2,\ldots,x_d)\in \R^{d-1}$ and $\chi\in C_c^{\infty}(\R,[0,1])$ is supported on $[-c,c]$ and equals $1$ on $[-c/2,c/2]$ for some small $c>0$ to be determined later. We then set
\begin{align}\label{def. regularized u}
u(x):=\sin(x_1)(\psi(x)-\kappa w(x)),\quad \kappa:=\left[(-1)^{\frac{m+1}{2}}\partial_1^mT(e_1)\right]^{-1}.
\end{align}
By \eqref{eq. derivative bounds f} we have
\begin{align}\label{lower bound for u}
|u(x)|\geq |\sin(x_1)|\psi(x)(1-C\kappa c^m)\gtrsim|\sin(x_1)|\psi(x)\gtrsim\dist(x_1,\pi\Z)\psi(x)
\end{align}
if $c$ is chosen sufficintly small. We now fix $c<\pi/10$ such that \eqref{lower bound for u} holds. This also guarantees that the functions $\chi(\cdot-k\pi)$, $k\in\Z$ have mutually disjoint supports. A straightforward computation yields 
\begin{align*}
(T(D)-\lambda) u(x)&=f(x)
-\frac{\kappa}{2\I}\sum_{|\alpha|\geq 1}\frac{1}{\alpha!}(e^{\I x_1}-(-1)^{|\alpha|}e^{-\I x_1})\partial^{\alpha}T(e_1)D^{\alpha} w(x)\\
&=f(x)-\kappa\sin(x_1)\sum_{\alpha\in 2\N}\frac{1}{\alpha!}\partial^{\alpha}T(e_1)(-1)^{\frac{|\alpha|}{2}}\partial^{\alpha} w(x)\\
&-\kappa\cos(x_1)\sum_{\alpha\in 2\N-1}\frac{1}{\alpha!}\partial^{\alpha}T(e_1)(-1)^{\frac{|\alpha|+1}{2}}\partial^{\alpha} w(x).
\end{align*}
This together with \eqref{eq. derivative bounds f}, \eqref{def. w} implies that
\begin{align}\label{bound on (T-lambda) u}
|(T(D)-\lambda) u(x)|\lesssim (1+|P_{\nu}x|+|P_{\nu}^{\perp}x|^2)^{-1}\psi(x).
\end{align}
Using \eqref{technical condition real potentials} we check that, for $k\in\Z$,
\begin{align*}
&(T(D)-\lambda) u(k\pi,x')=f(k\pi,x')\\
&-\kappa
\sum_{|\alpha|\in 2\N-1}\frac{1}{\alpha!}\partial^{\alpha}T(e_1)(-1)^{\frac{|\alpha|+1}{2}}\partial^{\alpha}((x_1-k\pi)^m\chi(x_1-k\pi)f(k\pi,x'))|_{x_1=k\pi}\\
&=f(k\pi,x')(1-\kappa(-1)^{\frac{m+1}{2}}\partial_1^mT(e_1))=0.
\end{align*}
To see this, split the sum into a part where $\alpha=me_1$ and its complement. The first part gives exactly the expression in the last line. The complementary part can be further split into three subparts where $\alpha_1<m$, $\alpha_1>m$ or $\alpha_1=m$. The first subpart is zero because there are too few $x_1$-derivatives that can fall onto the monomial $(x_1-k\pi)^m$. The second subpart is zero because there is at least one $x_1$-derivative that must fall onto the cutoff function $\chi$, and this is constant near $x_1=k\pi$. The third subpart is zero by assumption \eqref{technical condition real potentials}. 

Combining the result of the previous computation with \eqref{bound on (T-lambda) u} we get
\begin{align*}
|(T(D)-\lambda) u(x)|\lesssim (1+|P_{\nu}x|+|P_{\nu}^{\perp}x|^2)^{-1}\dist(x_1,\pi\Z)\psi(x).
\end{align*}
This together with \eqref{lower bound for u} yields that
\begin{align*}
|V(x)|=\frac{|(T(D)-\lambda) u(x)|}{|u(x)|}\lesssim (1+|P_{\nu}x|+|P_{\nu}^{\perp}x|^2)^{-1}.
\end{align*}
That $V$ is real-valued follows from Lemma \ref{lemma time reversal symmetry}.  
\end{proof}

\begin{corollary}\label{corollary real valued V}
Assume that the assumptions of Theorem~\ref{thm. Complex-valued potentials} or Theorem~\ref{thm. Complex-valued potentials II} hold and that $T$ is a radial polynomial. Then the conclusion of Theorem \ref{thm. Complex-valued potentials} or Theorem \ref{thm. Complex-valued potentials II} holds with real-valued potentials. 
\end{corollary}

\begin{proof}
The assumptions imply that $T$ is of the form $T(\xi)=T_0(|\xi|)$ where 
\begin{align*}
T_0(r)=\sum_{j=0}^Kc_lr^{2j}
\end{align*}
for some $K\in\N$ and $(c_l)_{l=1}^K\subset\R$. If $\eta\in M_{\lambda}$, then $|\eta|e_1\in M_{\lambda}$. The linear bijection $L\xi=|\eta|\xi$ preserves spherical symmetry, and $T_L(e_1)=T(Le_1)=\lambda$. 
We assume, as we may, that $K\neq 0$. Then 
\begin{align*}
|\xi+e_1|^{2K}=(1+2\xi_1+\xi_1^2+\ldots+x_d^2)^K=\xi_1^{2K}+2K\xi_1^{2K-1}+Q_K(\xi)
\end{align*}
where $Q_K$ is a sum of monomials that do not contain the factor $\xi_1^{2K-1}$. Similarly, none of the monomials appearing in $|\xi+e_1|^{2l}$, $l\leq K-1$, contain this factor. Therefore, condition \eqref{technical condition real potentials} holds with $m=2K-1$.
\end{proof}

\section{Proofs of Theorems \ref{thm. Radial potentials}--\ref{thm. Dirac}}

We first state a generalization of Propositions \ref{proposition pseudodifferential mapping properties} and \ref{proposition symbols} to symbols $a(x,\xi)$ depending on~$x$. We restrict ourselves to the case
\begin{align}\label{rho gamma special case}
\rho(x)=(1+|x|^2)^{1/2},\quad\gamma=(1,\ldots,1);
\end{align}
more general cases will not be needed. The proof is a straightforward adaptation of that of Propositions \ref{proposition pseudodifferential mapping properties} and \ref{proposition symbols} and will be omitted.

\begin{lemma}\label{lemma psdos radial weights}
Assume that $a\in C^{\infty}(\R^d\times\R^d)$ and that
\begin{align*}
|\partial_x^{\alpha}\partial_{\xi}^{\beta}a(x,\xi)|\leq C_{\alpha,\beta}(1+|x|)^{k-|\alpha|}(1+|\xi|)^{m_{\alpha,\beta}}.
\end{align*}
Then $a(x,h D):S_{\gamma}^{\ell}(\R^d)\to S_{\gamma}^{\ell+k}(\R^d)$ as a continuous map, with seminorms bounded independent of $h\in (0,1]$. Moreover, if $a(x,0)=0$ for all $x\in\R^d$, then $a(x,h D):S_{\gamma}^{\ell}(\R^d)\to hS_{\gamma}^{\ell+k-1}(\R^d)$.
\end{lemma}

\begin{proof}[Proof of Theorem \ref{thm. Radial potentials}]
We will prove the claim for $n=1$; the general case follows from a scaling argument similar to that in the proofs of Theorems \ref{thm. Complex-valued potentials}--\ref{thm. Complex-valued potentials II}.

Let $\sigma_{\lambda}$ be the uniform surface measure of $M_{\lambda}$. Since $T$ is radial, the Fermi surface $M_{\lambda}$ is a sphere and hence $\sigma^{\vee}_{\lambda}$ is radial. We fix $r_0\gg 1$ and pick a smooth radial function $\varphi$ that is positive for $|x|\leq r_0$ and equal to $\sigma^{\vee}_{\lambda}$ for $|x|>r_0$. We then define the radial function
\begin{align*}
\widetilde{u}(x):=\varphi(x)\psi(x),
\end{align*}
where $\psi(x):=(1+|x|^2)^{-N/2}$, with $N>d/2$.
A straightforward calculation using Fubini and the convolution theorem yields that, for $|x|>r_0$,
\begin{align*}
f(x):=(T(D)-\lambda)\widetilde{u}(x)=a(x,D)\psi(x)
\end{align*}
where 
\begin{align*}
a(x,\xi)=(2\pi)^{-d}\int_{M_{\lambda}}\e^{\I \eta\cdot x}(T(\xi+\eta)-\lambda)\rd\sigma_{\lambda}(\eta).
\end{align*}
Since $M_{\lambda}$ is compact, it follows from stationary phase (see e.g.\ \cite[(1.2.8)]{MR1205579}) that
\begin{align}\label{stationary phase estimate for b}
|\partial_x^{\alpha}\partial_{\xi}^{\beta}a(x,\xi)|\leq C_{\alpha,\beta}(1+|x|)^{-\frac{d-1}{2}-|\alpha|}(1+|\xi|)^{m_{\alpha,\beta}}.
\end{align}
Since $a(x,0)=0$, Lemma \ref{lemma psdos radial weights} yields
\begin{align}\label{radial estimate}
|f(x)|\lesssim(1+|x|)^{-\frac{d-1}{2}-1}\psi(x).
\end{align}
After rescaling, we may assume that $M_{\lambda}$ is the unit sphere $S^{d-1}$.
From \eqref{asymptotic surface measure} it follows that
that the zeros of $\widetilde{u}$ are simple and the distance between consecutive zeros is uniformly bounded below by a constant $\delta>0$.
Let us denote these zeros by $Z:=\{r_k\}_{k=1}^{\infty}\subset(r_0,\infty)$.
Then we have that
\begin{align}\label{zeros of Bessel are simple}
|\widetilde{u}(x)|\gtrsim \dist(|x|,Z)(1+|x|)^{-\frac{d-1}{2}}\psi(x).
\end{align}
Pick $\chi\in C_c^{\infty}(\R,[0,1])$ as in the proof of Proposition \ref{proposition real potential}, i.e.\ $\chi$ is supported on $[-c,c]$ and equals $1$ on $[-c/2,c/2]$ for some sufficiently small $c>0$; in particular, we require $c<\delta/4$, so that the functions $\chi(\cdot-r_k)$ have mutually disjoint support. By assumption, $T(D)$ is a differential operator of the form
\begin{align*}
T(D)=\sum_{j=0}^mc_j(-\Delta)^j
\end{align*}
where $\{c_j\}_{j=0}^m\subset \R$, $c_m\neq 0$. We may clearly assume that $m\geq 1$. 
Define the radial functions
\begin{align*}
w(x):=\sum_{k=1}^{\infty}(|x|-r_k)^{2m}\chi(|x|-r_k)f(r_k),
\end{align*}  
and 
\begin{align*}
u(x):=\widetilde{u}(x)-\kappa w(x),\quad \kappa:=c_m^{-1}(2d)^{-m}.
\end{align*}
It follows from \eqref{radial estimate} and \eqref{zeros of Bessel are simple} that 
\begin{align}\label{radial 1}
|u(x)|\gtrsim \dist(|x|,Z)\psi(x)(1+|x|)^{-\frac{d-1}{2}}(1-\mathcal{O}(c^{2m-1})).\end{align}
Moreover, for $|x|=r_k$ we have
\begin{align*}
(T(D)-\lambda)u(x)=f_0(r_k)-\kappa c_m(2d)^mf_0(r_k)=0,
\end{align*}
where we used the notation $f(x)=f_0(|x|)$. This together with \eqref{radial estimate} yields 
\begin{align}\label{radial 2}
|(T(D)-\lambda)u(x)|\lesssim \dist(|x|,Z)\psi(x)(1+|x|)^{-\frac{d-1}{2}-1}.
\end{align}
Hence, for $c$ suffiently small, we have by \eqref{radial 1}--\eqref{radial 2}
\begin{align*}
|V(x)|=\frac{|(T(D)-\lambda)u(x)|}{|u(x)|}\lesssim (1+|x|)^{-1}.
\end{align*}
It is clear that $V$ is radial.
\end{proof}

\begin{proof}[Proof of Theorem \ref{thm Chandrasekhar-Herbst}]
Here we combine our previous results with ideas from \cite{MR3661408} and \cite{2015arXiv150401144F}. 
Observe that if $\kappa:=\sqrt{(\lambda+1)^2-1}$, then $\kappa e_d\in M_{\lambda}$.
As before, we may assume (after a rescaling) that $\kappa=1$ (and hence $\lambda=\sqrt{2}-1$).
First define
\begin{align*}
g(t)&:=\int_0^{t}\sin^2(y)\rd y,\\
w_n(x)&:=(n^2+|x'|^4+g(x_d)^2)^{-N/2},
\end{align*}
where $n\in \N$ and $N>(d+1)/4$. We then set
\begin{align*}
\phi_n(x)&:=(\sqrt{(D+e_d)^2+1}+\sqrt{(D-e_d)^2+1})w_n(x),\\
u_n(x)&:=\sin(x_d)\phi_n(x).
\end{align*}
The eigenvalue equation will hold if we set $V_n=-(T(D)-\lambda)u_n/u_n$, as usual. We have to prove that this is a smooth function decaying as in \eqref{decay V nonradial}. An elementary computation (compare \eqref{T(D) applied to sin}) yileds
\begin{align}\label{Chandrasekhar V}
V_n(x)=\e^{\I x_1}\frac{\partial_{x_d}w_n(x)}{\sin(x_d)\phi_n(x)}-\frac{(T(D-e_d)-\lambda)\phi_n(x)}{\phi_n(x)}.
\end{align}
This is fortunate since the term with the sine in the denominator now has a differential (i.e.\ local) operator in the nominator (similar to the situation in Proposition~\ref{proposition real potential}). The zeros of the denominator will be cancelled by the choice of~$w_n$. In fact, using that $1+g(x_d)\gtrsim |x_d|$, we have
\begin{align}\label{Chandrasekhar 1}
\left|\frac{\partial_{x_d}w_n(x)}{\sin(x_d)}\right|\lesssim (n+|x'|^2+|x_d|)^{-N/2-1}
\end{align}
for large enough $n$. Next, we would like to apply Proposition \ref{proposition symbols} to the denominator in the second term of $V_n$ with
\begin{align*}
a(\xi):=(T(\xi-e_d)-\lambda)(\sqrt{(\xi+e_d)^2+1}+\sqrt{(\xi-e_d)^2+1})
\end{align*}
and $f=w_n$. Notice that $a(0)=0$ by the assumption that $e_1\in M_{\lambda}$. Here we cannot use the full strength of Proposition \ref{proposition symbols} since $w_n$ is not a symbol. It only satisfies the weaker estimates
\begin{align*}
|w_n(x)|&\lesssim (n+|x'|^2+|x_d|)^{-N/2},\\
|\partial^{\alpha}w_n(x)|&\lesssim (n+|x'|^2+|x_d|)^{-N/2-1},\quad |\alpha|\geq 1.
\end{align*}
However, as is obvious from the proof of Proposition \ref{proposition symbols}, this still implies that
\begin{align}\label{Chandrasekhar 2}
|a(D)w_n(x)|\lesssim (n+|x'|^2+|x_d|)^{-N/2-1}.
\end{align}
It remains to prove that 
\begin{align}\label{Chandrasekhar 3}
|\phi_n(x)|\gtrsim (n+|x'|^2+|x_d|)^{-N/2}
\end{align}
since \eqref{Chandrasekhar V}--\eqref{Chandrasekhar 3} then imply the desired properties of $V_n$. We write
\begin{align*}
\phi_n(x)&=2\sqrt{2}w_n(x)+a(D)w_n(x),\\
a(\xi)&:=\sqrt{(\xi+e_d)^2+1}+\sqrt{(\xi-e_d)^2+1}-2\sqrt{2}.
\end{align*}
Since $a(0)=0$, the same modification of Proposition \ref{proposition symbols} as above yields
\begin{align*}
|\phi_n(x)|&\geq 2\sqrt{2}w_n(x)-\mathcal{O}(1)(n+|x'|^2+|x_d|)^{-N/2-1}\\
&\gtrsim (n+|x'|^2+|x_d|)^{-N/2}(1-\mathcal{O}(1/n))
\end{align*}
Hence \eqref{Chandrasekhar 3} will be satisfied for $n$ sufficiently large.
\end{proof}

\begin{proof}[Proof of Theorem \ref{thm. Dirac}]
Similarly as in the proof of Theorem \ref{thm Chandrasekhar-Herbst} we may assume witout loss of generality that $\lambda=\sqrt{2}$. Let $K$ be the dimension of the spinor space, i.e.\ $\alpha_j$ and $\beta$ are $K\times K$ matrices, and we are seeking an eigenfunction $u\in L^2(\R^d;\C^K)$. From the properties of the Dirac matrices it follows that the matrix $\alpha_d+\beta$ has eigenvalues $\sqrt{2}$ and $-\sqrt{2}$. We fix an eigenvector $\mathbf{v}\in \C^K$ corresponding to $\lambda=\sqrt{2}$ and define
\begin{align*}
u_n(x)&:=\e^{\I x_d}\psi_n(x)\mathbf{v},\\
\psi_n(x)&:=(n^2+|x'|^4+|x_d|^2)^{-N/2}
\end{align*}
with $n\in\N$ and $N>(d+1)/4$. We then compute
\begin{align*}
\mathcal{D}u(x)&=\e^{\I x_d}\psi_n(x)(\alpha_d+\beta)\mathbf{v}+\e^{\I x_d}\sum_{j=1}^dD_j\psi_n(x)\alpha_j\mathbf{v}\\
&=\lambda u_n(x)+\e^{\I x_d}\sum_{j=1}^dD_j\psi_n(x)\alpha_j\mathbf{v}.
\end{align*}
If we set
\begin{align}\label{V Dirac}
V_n(x):=-\frac{1}{\psi_n(x)}\sum_{j=1}^dD_j\psi_n(x)\alpha_j,
\end{align}
then we obtain
\begin{align*}
(\mathcal{D}+V_n(x)-\lambda)u_n(x)=\e^{\I x_d}\left(\sum_{j=1}^dD_j\psi_n(x)\alpha_j+V_n(x)\psi_n(x)\right)\mathbf{v}=0.
\end{align*}
Since $\alpha_j$ are hermitian matrices, $\psi_n$ is real-valued and $D_j=-\I\partial_j$, it follows that $V_n$ is anti-hermitian.
\end{proof}

\bibliographystyle{plain}
\bibliography{C:/Users/Jean-Claude/Dropbox/papers/bibliography_masterfile}

\begin{thebibliography}{10}

\bibitem{MR0276624}
S.~Agmon.
\newblock Lower bounds for solutions of {S}chr\"odinger equations.
\newblock {\em J. Analyse Math.}, 23:1--25, 1970.

\bibitem{MR3608659}
J.-C. Cuenin.
\newblock Eigenvalue bounds for {D}irac and fractional {S}chr\"odinger
  operators with complex potentials.
\newblock {\em J. Funct. Anal.}, 272(7):2987--3018, 2017.

\bibitem{MR2820160}
R.~L. Frank.
\newblock Eigenvalue bounds for {S}chr\"odinger operators with complex
  potentials.
\newblock {\em Bull. Lond. Math. Soc.}, 43(4):745--750, 2011.

\bibitem{2015arXiv150401144F}
R.~L. {Frank} and B.~{Simon}.
\newblock {Eigenvalue bounds for Schr\"odinger operators with complex
  potentials. II}.
\newblock {\em ArXiv e-prints}, April 2015.

\bibitem{MR3243741}
L.~Grafakos.
\newblock {\em Modern {F}ourier analysis}, volume 250 of {\em Graduate Texts in
  Mathematics}.
\newblock Springer, New York, third edition, 2014.

\bibitem{MR2024415}
A.~D. Ionescu and D.~Jerison.
\newblock On the absence of positive eigenvalues of {S}chr\"odinger operators
  with rough potentials.
\newblock {\em Geom. Funct. Anal.}, 13(5):1029--1081, 2003.

\bibitem{MR3449177}
Hiroshi Isozaki and Hisashi Morioka.
\newblock Inverse scattering at a fixed energy for discrete {S}chr\"odinger
  operators on the square lattice.
\newblock {\em Ann. Inst. Fourier (Grenoble)}, 65(3):1153--1200, 2015.

\bibitem{MR0108633}
T.~Kato.
\newblock Growth properties of solutions of the reduced wave equation with a
  variable coefficient.
\newblock {\em Comm. Pure Appl. Math.}, 12:403--425, 1959.

\bibitem{MR1809288}
C.~E. Kenig and N.~Nadirashvili.
\newblock A counterexample in unique continuation.
\newblock {\em Math. Res. Lett.}, 7(5-6):625--630, 2000.

\bibitem{MR1880829}
H.~Koch and D.~Tataru.
\newblock Sharp counterexamples in unique continuation for second order
  elliptic equations.
\newblock {\em J. Reine Angew. Math.}, 542:133--146, 2002.

\bibitem{KochTataru2006}
H.~Koch and D.~Tataru.
\newblock Carleman estimates and absence of embedded eigenvalues.
\newblock {\em Comm. Math. Phys.}, 267(2):419--449, 2006.

\bibitem{MR3661408}
J.~Lorinczi and I.~Sasaki.
\newblock Embedded eigenvalues and {N}eumann--{W}igner potentials for
  relativistic {S}chr\"odinger operators.
\newblock {\em J. Funct. Anal.}, 273(4):1548--1575, 2017.

\bibitem{MR3052498}
C.~Muscalu and W.~Schlag.
\newblock {\em Classical and multilinear harmonic analysis. {V}ol. {I}}, volume
  137 of {\em Cambridge Studies in Advanced Mathematics}.
\newblock Cambridge University Press, Cambridge, 2013.

\bibitem{MR0493421}
M.~Reed and B.~Simon.
\newblock {\em Methods of modern mathematical physics. {IV}. {A}nalysis of
  operators}.
\newblock Academic Press [Harcourt Brace Jovanovich, Publishers], New
  York-London, 1978.

\bibitem{MR0247300}
B.~Simon.
\newblock On positive eigenvalues of one-body {S}chr\"odinger operators.
\newblock {\em Comm. Pure Appl. Math.}, 22:531--538, 1969.

\bibitem{MR1205579}
C.~D. Sogge.
\newblock {\em Fourier integrals in classical analysis}, volume 105 of {\em
  Cambridge Tracts in Mathematics}.
\newblock Cambridge University Press, Cambridge, 1993.

\bibitem{MR1232192}
E.~M. Stein.
\newblock {\em Harmonic analysis: real-variable methods, orthogonality, and
  oscillatory integrals}, volume~43 of {\em Princeton Mathematical Series}.
\newblock Princeton University Press, Princeton, NJ, 1993.
\newblock With the assistance of Timothy S. Murphy, Monographs in Harmonic
  Analysis, III.

\bibitem{MR2827930}
Elias~M. Stein and Rami Shakarchi.
\newblock {\em Functional analysis}, volume~4 of {\em Princeton Lectures in
  Analysis}.
\newblock Princeton University Press, Princeton, NJ, 2011.
\newblock Introduction to further topics in analysis.

\bibitem{MR1219537}
B.~Thaller.
\newblock {\em The {D}irac equation}.
\newblock Texts and Monographs in Physics. Springer-Verlag, Berlin, 1992.

\bibitem{wigner1929merkwurdige}
EP~Wigner and J~Von-Neumann.
\newblock {\"U}ber merkw{\"u}rdige eigenwerte.
\newblock {\em Z. Phys}, 30:465, 1929.

\end{thebibliography}
\end{document}